\documentclass{amsart}
\usepackage{latexsym,amsfonts}
\usepackage{tikz}
\usetikzlibrary{calc,shapes,arrows,automata}
\usepackage{amsmath}
\usepackage{algorithm}
\usepackage{algorithmic}
\usepackage{url}
\usepackage{subcaption} 
\usepackage{caption}
\usepackage{centernot}
\usepackage{svg}
\usepackage{epsfig}

\usepackage{thmtools}
\usepackage{graphicx} 
\newtheorem{theorem}{Theorem}
\newtheorem{lemma}{Lemma}
\newtheorem{definition}{Definition}[section]
\usepackage{bm}
\usepackage{pifont}

\newcommand{\Reals}{\mathbb{R}}
\newcommand{\Nats}{\mathbb{N}}
\newcommand{\Realpos}{\mathbb{R}}
\newcommand{\system}{\mathcal{S}}
\newcommand{\states}{\mathcal{X}}
\newcommand{\initstates}{\mathcal{X}_{0}}
\newcommand{\secretstates}{\mathcal{X}_{s}}

\newcommand{\transfunc}{f}

\newcommand{\transclos}{\mathcal{T}}
\newcommand{\barrier}{\mathcal{B}}
\newcommand{\lab}{\mathcal{L}}
\newcommand{\AP}{\mathsf{AP}}
\newcommand{\Traces}{\mathcal{T}}
\newcommand{\until}{\mathsf{U}}
\newcommand{\release}{\mathsf{R}}
\newcommand{\nex}{\mathsf{X}}
\newcommand{\always}{\mathsf{G}}
\newcommand{\eventually}{\mathsf{F}}
\newcommand{\Aut}{\mathcal{A}}
\newcommand{\Lang}{\mathsf{L}}
\newcommand{\Inf}{\text{Inf}}
\newcommand{\prodsys}{\bm{\system}}
\newcommand{\prodstate}{\bm{\states}}
\newcommand{\prodinitstate}{ \bm{ \mathcal{X}}_0 }
\newcommand{\prodtrans}{\bm{\transfunc}}
\newcommand{\InPath}{\text{InPath}}
\newcommand{\indexforall}{\mathfrak{F}}
\newcommand{\indexexists}{\mathfrak{E}}

\newtheorem{problem}{Problem}

\newcommand{\quant}{\mathsf{Q}}
\newcommand{\xmark}{\ding{55}}

\usepackage{tikz}
\usetikzlibrary{automata, positioning}

\title[HyperCertificates]{HyperCertificates: \\ Verification of Discrete-time Dynamical Systems against HyperLTL Specifications}
\author{Vishnu Murali, Amin Falah, 
Ashutosh Trivedi, and
Majid Zamani
}
\thanks{
This work was supported by NSF under grants CNS-2111688 and CNS-2145184.} 
\thanks{ 
 V. Murali, A. Falah, A. Trivedi and M. Zamani are with the Department of Computer Science, University of Colorado Boulder, CO, USA. Emails: \{vishnu.murali,amin.falah,ashutosh.trivedi,majid.zamani\}@colorado.edu.}
\begin{document}
\begin{abstract}
We introduce a functional inductive framework to verify discrete-time
dynamical systems against hyperproperties specified as Hyperlinear
temporal logic formulae via a notion of HyperCertificates. Unlike
linear temporal logic (LTL) formulae which are concerned with
individual traces of a system, hyperproperties are properties that are
concerned with how the traces of a system relate to one another.
HyperLTL is an extension of LTL for hyperproperties, and is useful to
describe specifications such as opacity, privacy as well as notions of
robustness. Our notion of HyperCertificates consists of a pair of functions, where the first models the lookahead, and the second relies on a combination of barrier and ranking functions. We use closure certificates, to act as a model for this lookahead and then rely on barrier and ranking function arguments modulo this lookahead to provide guarantees against HyperLTL formulae. We demonstrate how our approach is automatable via existing techniques such as sum-of-squares	optimization (SOS) and satisfiability modulo theories (SMT) solvers. Finally, we demonstrate our approach on some case studies.
\end{abstract}

\maketitle

\maketitle

\section{Introduction} \label{sec:intro}
Temporal logic properties, characterized in Linear Temporal Logic (LTL) \cite{pnueli_1977_temporal} or as $\omega$-automata \cite{vardi_2005_automata}, typically capture properties over individual traces of a system.
Examples of such properties include safety, where one ensures that a system does not reach some unsafe states, and reachability, where one ensures that some ``good'' states are eventually reached from some initial set of states.
To reason about the satisfaction of such properties, one may reason about individual traces of a system independent of one another.
Unfortunately, such temporal formulae cannot capture many complex properties that describe how traces of a system \emph{relate} to each other.
Such properties include information flow properties concerned with plausible deniability and security, such as opacity \cite{Bryan_2006_opacity} and noninterference \cite{Goguen_Security_nonint_1982}. 
Similarly, the notions of optimality and robustness \cite{wang2020hyperproperties}, which are useful in the planning of robotic tasks, require reasoning about how the traces of a system relate to each other.
To tackle this issue, the authors of \cite{clarkson_hyperproperties_2010} introduced the notion of hyperproperties and characterized them using temporal hyperlogics \cite{clarkson_temporal_2014} (analogs of classical temporal logics) such as HyperLTL.
This paper considers the problem of verifying discrete-time dynamical systems against HyperLTL formulae via a notion of HyperCertificates.

\noindent \textbf{Verification of Hyperproperties.} 
Model checking algorithms for  hyperlogics have been well studied \cite{finkbeiner2015algorithms,finkbeiner2021model} for  finite-state systems.
However, the verification of infinite-state systems against such properties remains an open problem.
The authors of \cite{finkbeiner2023automata} consider a model checking procedure where they overapproximate an infinite-state system with a finite system and prune this overapproximation via a notion of $k$-step infeasibility. They then verify this property by allowing the existential player to select a strategy based on a $k$-step lookahead where $k$ is a user-defined fixed parameter.
However, the implementation of the model checking algorithm is left as future work.
The authors of \cite{correnson2025coinductive} consider coinductive proof rules for the deductive verification of infinite-state programs against hyperproperties in the $\forall^{*} \exists^{*}$-fragment of HyperLTL.
Unfortunately, these rules are based on a \emph{manual} proof in the Coq interactive theorem prover \cite{Coq-refman} and are not easily automatable.

\noindent \textbf{Augmented Barrier Certificates.} The authors of \cite{anand2024verification} consider an approach to verify uncertain dynamical systems against HyperLTL formulae via the notion of augmented barrier certificates.
Although this is automated, the approach in general is conservative (cf. Section \ref{sec:abc_hyperLTL}) and only works for some safety specifications.
We now elaborate on some of these drawbacks due to the absence of foresight. 
Typically, verifying HyperLTL formulae can be thought of as a two-player game between the universal ($\forall$) and the existential ($\exists$) player, respectively. 
Here, the $\forall$-player selects traces to falsify the property, whereas the $\exists$-player attempts to satisfy it by choosing  corresponding counter-traces.
If the $\exists$-player wins, then the system satisfies the desired property.
As discrete-time dynamical systems are typically continuous-state, one often fails to find a finite representation of the traces of the system (that is bisimilar), and so one cannot directly consider this trace-based game.
Thus, augmented barrier certificates model this two-player game in a turn-based fashion rather than a trace-based fashion.
Here,  the $\forall$-player attempts to falsify the property by selecting the immediate next state in the universal trace, and the $\exists$-player selects the next state in the counter-trace.
In this framework, the strategy of the $\exists$-player is restricted to reasoning about only the immediate next step, without access to possible future choices of the $\forall$-player.
To illustrate this limitation, consider the hyperproperty of the initial-state opacity \cite{lafortune2018history}.
A system is said to be initial-state opaque if, for every trace that originates from a secret initial state, there exists a trace from a non-secret initial state that produces the same sequence of observable outputs. We use this property as a running example to highlight the drawbacks of augmented barrier certificates.

\begin{figure}[h]
   \begin{tikzpicture}
    \begin{scope}[xshift=0cm] 
    \node[initial, state, initial text=, fill=blue!40!white] at (0,0) (0) {$x_0$};
    \node[below=0.05cm of 0] {$s$};
     \node[state,fill=blue!40!white ] at (2,1) (1) {$x_1$};
     \node[state,fill=blue!40!white ] at (4,1) (2) {$x_2$};
     \node[state,fill=blue!40!white ] at (2,-1) (3) {$x_3$};
     \node[state,fill=red!40!white ] at (4,-1) (4) {$x_4$};
    \node[initial,initial text=, state,fill=blue!40!white ] at (0,-3) (5) {$x_5$};
    \node[below=0.05cm of 5] {$\neg s$};
    \node[state,fill=blue!40!white ] at (2,-3) (6) {$x_6$};
     \node[state,fill=blue!40!white ] at (4,-3) (7) {$x_7$};
     
     \node[initial,initial text=, state,fill=blue!40!white ] at (0,-5) (8) {$x_8$};
     \node[below=0.05cm of 8] {$\neg s$};
     \node[state,fill=blue!40!white ] at (2,-5) (9) {$x_9$};
     \node[state,fill=red!40!white ] at (4,-5) (10) {$x_{10}$};

     \path[->]
     (0)  edge node{} (1)
     (0)  edge node{} (3)
     (1)  edge node{} (2)
     (3)  edge node{} (4)
     (2) edge[loop right] node{} (2)
     (4) edge[loop right] node{} (4)
     
     (5)  edge node{} (6)
     (8)  edge node{} (9)
     (6)  edge node{} (7)
     (9)  edge node{} (10)
     (7) edge[loop right] node{} (7)
     (10) edge[loop right] node{} (10);
    \end{scope}
   \end{tikzpicture}

  \caption{An example illustrating the limitations of augmented barrier certificates in verifying opacity. The discrete time system demonstrates a case where one fails to find an augmented barrier to establish opacity, even though the system is opaque.}
 \label{fig:abc_fails}
\end{figure}
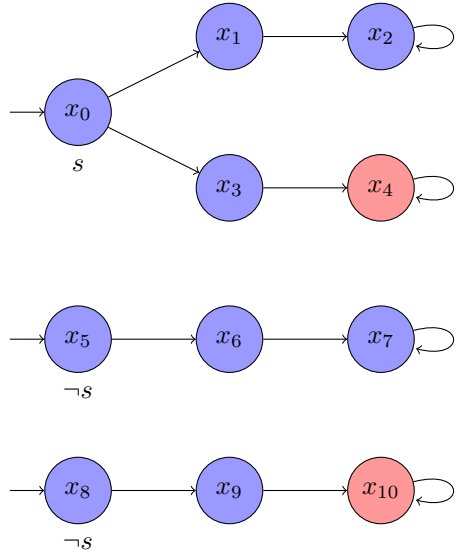

Consider a non-deterministic finite state system in Figure~\ref{fig:abc_fails} where the color of each state is the observed output.
Here, the initial secret state $x_0$ may transition to either $x_1$ or $x_3$, resulting in output sequences $(\text{blue}, \text{blue}, \text{blue}, \text{blue}^{\omega} )$ (an infinite sequence of blue) or $(\text{blue}, \text{blue}, \text{red}, \text{red}^{\omega})$, respectively. 
Observe that this system is opaque, as we can choose to start from either $x_5$ or $x_8$ (initial non-secret states) to match these output sequences. 
Unfortunately, under the game-based interpretation of augmented barrier certificates, the $\forall$-player may choose either trace from $x_0$, and without foresight into that choice, the $\exists$-player cannot guarantee a correct counter-trace. 
This results in a failure of the augmented barrier certificates, which operate locally and cannot reason over the future behavior of the opposing trace.

\noindent \textbf{Closure Certificates.} To overcome this, we introduce a \emph{lookahead variable} that encodes a future state that could be visited by the $\forall$-player.
This lookahead augments the barrier function, allowing the $\exists$-player to reason about the outcome of the future choice of the $\forall$-player and thus helps in the selection of a matching counter-trace. 
In the example above, if the trace from $x_0$ leads to $x_4$, the lookahead facilitates choosing $x_8$ as a suitable initial non-secret state, thus restoring the ability to certify opacity.
To determine which lookahead states are relevant, we leverage \emph{closure certificates} \cite{murali2024closure}, which provide a sound overapproximation of the system’s reachable transitions.
By induction, a closure certificate  overapproximates reachable pairs of states, enabling us to compactly represent the reachability relation across traces. This contrasts with barrier certificates, which overapproximate reachable \emph{states} but do not encode pairwise reachability between states across steps
(cf. Section \ref{subsec:ac_opacity_1}).
By combining the notions of closure certificates for reasoning about reachability across traces and barrier certificates for state-based separation, we arrive at a more expressive framework for verifying hyperproperties. 
We summarize our main contributions below. 

\noindent \textbf{Contributions.}
(i) We present a notion of HyperCertificates to verify systems against HyperLTL formulae. 
(ii) We demonstrate how to effectively automate the search for these certificates using many of the standard techniques that have been used to verify the $\omega$-regular properties of interest. (iii) We compare and and contrast our approach with the capabilities of existing approaches through some case studies.

\noindent \textbf{Related works.} 
The authors of \cite{clarkson_hyperproperties_2010} introduced the notion of hyperproperties to capture properties that describe how the traces of a system relate to each other. 
These properties have been captured through hyperlogics such as Hyper linear temporal logic (HyperLTL), hyper computation tree logic (HyperCTL) \cite{clarkson_temporal_2014}, and hyper signal temporal logic (HyperSTL) \cite{nguyen2017hyperproperties}. 
Examples of such properties include ensuring the secrecy of sensitive information \cite{zdancewic2003observational,alur2006preserving,berard2015probabilistic} such as plausible deniability via opacity \cite{berard2015probabilistic,jacob2016overview,lafortune2018history}, as well as for other objectives based on robustness and optimization \cite{wang2020hyperproperties}.
The results in \cite{finkbeiner2015algorithms} consider algorithms for model checking finite-state systems against HyperLTL and HyperCTL properties.
The results in \cite{finkbeiner2019monitoring} consider 
the verification and monitoring of the hyperproperties at runtime.
For the verification of infinite state systems, the results in \cite{finkbeiner2023automata} consider the abstraction of the original system to a program automaton and provide a guaranty with respect to this abstraction.
The results in \cite{correnson2025coinductive} consider the use of coinductive proof rules to verify hyperliveness properties, although such proofs are not automated.
Inspired by the use of barrier certificates \cite{prajna_2004_safety} as an automated search for invariants, the authors of \cite{liu2020verification} consider the notion of augmented barrier certificates to ensure opacity.
Based on this, the results in \cite{anand2024verification} consider an extension of augmented barrier certificates to verify HyperLTL properties.
Analogously to barrier certificates, the results in \cite{murali2024closure} consider a notion of closure certificates to automate the search for transition invariants \cite{podelski_2004_transition} of real-valued dynamical systems.
Inspired by this, we introduce HyperCertificates that combine barrier and closure certificates to verify HyperLTL properties.
We summarize a comparison of our work with the existing results in Table \ref{tab:comparison}.

\begin{table}
\resizebox{1.3\textwidth}{!}{
\begin{tabular}{|c|c|c|c|c|c|c|}
    \hline
         Approach&HyperLTL-fragment&System size&Automation of Proof& Lookahead&Distance to lookahead&Necessity\\
        \hline
         Our work (HyperCertificates)&All& Infinite&Automated&\checkmark&arbitrary&\xmark\\
         \hline 
         PLDI'25 \cite{correnson2025coinductive}&$\forall^* \exists^*$&  Infinite&Manual&\checkmark&arbitrary&\checkmark$^{\Diamond}$\\
         \hline 
        TAC'24 \cite{anand2024verification}&All&  Infinite&Automated&\xmark&$0$&\xmark\\
        \hline 
           NFM'23 \cite{finkbeiner2023automata}&$\forall^* \exists^*$&  Infinite&Automated&\checkmark&fixed length $k$&\xmark\\
        \hline       
CAV'15 \cite{finkbeiner2015algorithms}&All&Finite&Automated&\checkmark&arbitrary&\checkmark\\
\hline
    \end{tabular}
}
    \caption{A table comparing our work with some of the existing results in the verification of HyperLTL. The results of \cite{correnson2025coinductive} are not necessary for $\omega$-lookaheads but are necessary when considering finite lookaheads.}
    \label{tab:comparison}
\end{table}

\section{Preliminaries}
\label{sec:prelims}
\subsection{Notation}
We use $\Nats$ and $\Reals$ to denote the set of natural numbers (including zero) and reals, respectively. 
Given an element $a \in \Reals$, we use $\Reals_{\geq a}$ (resp. $\Nats_{\geq a}$) and $\Reals_{ > a}$  to denote the sets $[a, \infty[$ and $]a,\infty[$, respectively.
Given a set $A$, we use $A^*$ to denote the set of finite-length sequences of elements from $A$, and $A^{\omega}$ to denote the set of countably infinite-length sequences, and let $A^{\infty} = A^* \cup A^{\omega}$ denote the set of all sequences (both finite and countably infinite).
We use $s = \langle a_0, a_1, \ldots, a_n \rangle$ to denote a finite sequence and $s = \langle a_0, a_1, \ldots \rangle$ to denote an infinite sequence.
Given a sequence $s = \langle a_0, a_1, \ldots \rangle \in A^{\infty} $ and two indices $i \in \Nats$, $j \in \Nats \cup \{\infty\}$, we define the sequence $s[i,j] = \langle a_i, a_{i+1}, \ldots, a_{j} \rangle$ if $j \in \Nats$ and $s[i,j] = \langle a_i, a_{i+1}, \ldots \rangle$ otherwise.
Similarly, given a sequence $s = \langle a_0, a_1, \ldots \rangle \in A^{\infty} $ and an index $i \in \Nats$, the term $s[i]$ represents the element $a_i \in A$.
Given an infinite sequence $s = \langle a_0, a_1, \ldots \rangle$, we use $\Inf(s)$ to denote the set of elements that occur infinitely often, \textit{i.e.}, we have $a \in \Inf(s)$, if for all $i \in \Nats$, there exists $j > i$ such that $s[j] = a$.
Given a relation $R \subseteq A \times B$ and an element $a \in A$, we use $R(a)$ to denote the set $\{ b \mid (a,b) \in R \}$. 
Given a set $A$ and a natural number $m$, we use $A^m$ to denote the $m$-fold product of $A$ by itself.
We use $\bm{a}  =(a_1, \ldots , a_m) \in A^m $ to denote a vector of elements in $A$.

 \subsection{Discrete-Time Systems}
A discrete-time system (simply a system) is a tuple of the form $ \system = (\states, \initstates, \transfunc) $, where $ \states $ denotes the set of states, $ \initstates \subset \states $ is the set of initial states and $ \transfunc: \states \rightarrow 2^{\states} $ is the transition map.
The evolution of the system is described by the transition relation $f$.
A \emph{state sequence} is defined as an infinite sequence $ \langle x_0, x_1, x_2, \ldots \rangle \in \states^{\omega} $, where $ x_{i+1} \in \transfunc(x_{i}) $ and $ x_0 \in \initstates $ for all $i \in \Nats$. 
In contrast, a finite state sequence is represented as $ \langle x_0, x_1, \ldots, x_n \rangle $ for some $ n \in \Nats $.
With a slight abuse of notation, we use $f: \states \rightarrow \states$ when the transition map is a function.
In this work, we use the terms state sequence and path interchangeably.

\subsection{Barrier Certificates and Safety}
\label{subsec:safety_bar}
A system $\system = (\states, \initstates, f)$ is safe with respect to a set of unsafe states $\states_U \subseteq \states$ if for all state sequences $\langle x_0, x_1, \ldots \rangle$ we have $x_i \notin \states_u$ for all $i \in \Nats$.
To ensure safety, one may make use of the notion of barrier certificates \cite{prajna_2004_safety}.
\begin{definition}
\label{def:bar}
    Consider a system $\system = (\states, \initstates, f)$ and a set of unsafe states $\states_U \subseteq \states$.
    A function $\barrier: \states \to \Reals$ is a barrier certificate if:
    \begin{align}
        & \barrier(x_0) \geq 0 && \text{ for all } x_0 \in \initstates, \label{eq:bar_init}\\
        & \barrier(x_u) < 0 && \text{ for all } x_u \in \states_u, \text{ and } \label{eq:bar_unsafe} \\
        & \big( \barrier(x) \geq 0 \big)  \nonumber\\ &\implies \big(B(x') \geq 0 \big) && \text{ for all } x \in \states, x' \in f(x). \label{eq:bar_inv}
    \end{align}
\end{definition}
The existence of a function $\barrier$ as in Definition \ref{def:bar} ensures that a system $\system$ is safe \cite{prajna_2004_safety}.
We should note that from our definition we consider the invariant to be characterized as $\{ x \mid \barrier(x) \geq 0 \}$ unlike traditional barrier approaches where one considers the set $\{ x \mid \barrier(x) \leq  0 \}$ to be an invariant.
This is to ensure consistency with the definition of closure certificates.

\subsection{Finite Visits and Closure Certificates}
\label{subsec:fin_vis_clos}
A system $\system = (\states, \initstates, f)$ visits a set of states $\states_{VF} \subseteq \states$ only finitely often, if for all state sequences $\langle x_0, x_1, \ldots \rangle$ there exists some $j \in \Nats$ such that for all $i > j$, we have $x_i \notin \states_{VF}$.
To ensure that a region is visited only finitely often, one may make use of the notion of closure certificates \cite{murali2024closure}.
\begin{definition}
\label{def:cc}
    Consider a system $\system = (\states, \initstates, f)$.
    A function $\transclos: \states \times \states \to \Reals$ is a closure certificate if for all $ x,y \in \states$ and $x' \in f(x)$  we have:
    \begin{align}
        & \transclos(x, x') \geq 0, \text{ and} \label{eq:cc_base}\\
        & \big( \transclos(x', y) \geq 0 \big) \implies \big( \transclos(x, y) \geq 0 \big).\label{eq:cc_inv}
    \end{align}
\end{definition}
Given a set $\states_{VF}$ that must be visited finitely often (a.k.a. persistence specification), a closure certificate for persistence consists of a function $\transclos$ that satisfies conditions \eqref{eq:cc_base} and \eqref{eq:cc_inv}, and ensures that there exists $\xi \in \Reals_{ > 0}$, such that for all $x_0, \in \initstates$, $y',y'' \in \states_{VF}$,
    \begin{align}
        &\big( \transclos(x_0, y') \geq 0 \big) \wedge \big( \transclos(y',y'') \geq 0 \big) \nonumber\\ &\implies \big(\transclos(x_0,y'')\leq \transclos(x_0,y') - \xi). \label{eq:cc_decrease}
    \end{align}

The existence of a closure certificate for persistence ensures that a system visits the set $\states_{VF}$ only finitely often \cite{murali2024closure}.

\subsection{Linear Temporal Logic and B\"uchi Automata}
The Formulae in LTL~\cite{pnueli_1977_temporal}  are defined with respect to a set of finite atomic propositions $\AP$ that are relevant to our system. 
Let $\Sigma = 2^{\AP}$ denote the power set of atomic propositions.
A trace, or word, $w = \langle w_0, w_1, \ldots, \rangle \in \Sigma^{\omega}$ is an infinite sequence of sets of atomic propositions.
The syntax of LTL is described by the following grammar: 
\begin{align}
\psi &:= \top \;|\ a \;|\; \neg \psi \;|\; \psi \wedge \psi \;|\;  \nex \psi  \;|\; \psi \until \psi, \nonumber
\end{align}
where $\top$ indicates true, $a \in AP$ denotes an atomic proposition, and $\wedge$, $\neg$ denote the logical conjunction and negation, respectively.
The temporal operators next, and until, are denoted by $\nex$, and $\until$, respectively.
One can derive the other operators such as disjunction ($\lor$), release ($\release$), eventually ($\eventually$) and always ($\always$) respectively.
We refer the interested reader to \cite{vardi_2005_automata} for more details on the semantics.
To connect the state sequences of a system with a specification characterized in LTL, we associate a labeling function $\lab: \states \to \Sigma$ which maps each state of the system to a letter in the finite alphabet $\Sigma$.
That is, each state of the system is assigned a set of (possibly empty) atomic propositions.
This naturally generalizes to mapping a state sequence of the system $\langle x_0, x_1, \ldots \rangle \in \states^{\omega}$ to a trace $w = \langle \lab(x_0), \lab(x_1), \ldots \rangle \in \Sigma^{\omega}$. 
We denote the set of all the traces of a system $\system$ under the labeling map $\lab$ by $\Traces_{\lab}(\system)$.
A system $\system$ satisfies an LTL property $\phi$ under the labeling map $\lab$ if for all $w \in \Traces_{\lab}(\system)$, we have $w \models \psi$.
We denote this by $\system \models_{\lab} \psi$ and infer the labeling map from context.

Following \cite{vardi_2005_automata}, every LTL specification $\psi$ can be compiled into a nondeterminstic B\"uchi automaton (NBA) $\Aut$.
A nondeterminstic B\"uchi automaton (NBA) is a tuple $\Aut = (\Sigma,Q,q_0,\delta,Q_{Acc}) $, where $\Sigma$ is a finite alphabet, $Q$ is a finite set of states, $q_0 \in Q_0$ is the initial state of the automaton, $\delta \subseteq Q \times \Sigma \times Q$ is the transition relation, and $Q_{Acc} \subseteq Q$ denotes a set of accepting states.
A run of an NBA $\Aut = (\Sigma,Q,q_0, \delta, Q_{Acc})$ over a word $w = \langle w_0, w_1, \ldots \rangle \in {\Sigma}^{\omega}$ is an infinite sequence of states $\rho = \langle q_0, q_1, \ldots \rangle$ such that $q_0$ is the initial state, and $(q_{i}, w_i, q_{i+1}) \in \delta$ for all $i \in \Nats$.
A word $w \in \Sigma^{\omega}$ is said to be accepted by NBA $\Aut$ if there exists a run $\rho = \langle q_0, q_1, \ldots \rangle$ such that $\Inf(\rho) \cap Q_{Acc} \neq \emptyset$.
The language of an NBA $\Aut$ denoted by $\Lang(\Aut) \subseteq \Sigma^{\omega}$ is defined as $\Lang(\Aut) = \{ w \mid w  \text{ is accepted by} \Aut \}$.
Given an LTL formula $\psi$ over a set of atomic propositions $\AP$, there exists an NBA $\Aut = (\Sigma, Q, q_0, \delta,Q_{Acc})$ where $\Sigma = 2^\AP$ such that for all traces $w\in \Sigma^{\omega}$, we have $w \in \Lang(\Aut)$ if and only if $w \models \psi$.
We now consider the extension of LTL to hyperlogics.

\subsection{HyperLTL Specifications}
\label{subsec:HyperLTL}
Similarly to LTL formulae, HyperLTL formulae \cite{clarkson_hyperproperties_2010} are defined with respect to a set of atomic propositions $\AP$ that are relevant to the system.
Unlike LTL, HyperLTL can express specifications that are concerned with how traces of a system are related to each other.
These formulae can be thought of as a combination of trace variables and quantification with LTL.
The syntax of HyperLTL is described by the following grammar:
\begin{align}
\phi &:= \forall \pi\phi \;|\; \exists \pi \phi \;|\; \psi, \nonumber \\
\psi &:= \top \;|\; a_{\pi} \;|\; \psi \wedge \psi \;|\; \neg \psi \;|\; \psi \until \psi. \nonumber
\end{align}
The formulae generated by $\phi$ are related to quantification of traces of the system.
The formula $\forall \pi \phi$ states that for all traces $\pi$ of the system, the formula $\phi$ is valid.
Similarly, formula $\exists \pi \phi$ indicates that for some trace $\pi$ of the system the formula $\phi$ is valid.
Observe that these formulae may be nested, \textit{i.e.}, formula $\forall \pi_1 \exists \pi_2 \phi$ indicates that for all traces $\pi_1$ of the system, there exists some trace $\pi_2$ such that formula $\phi$ holds.
As these formulae quantify over traces, the variables $\pi$ are considered as trace variables.
The formulae generated by $\psi$ can be thought of as standard LTL formulae. 
The difference is that we use $a_{\pi}$ to denote that the atomic proposition $a$ holds for the trace variable $\pi$.
A trace variable is said to be free in a given HyperLTL formula if it is not quantified, \textit{e.g.}, $ \forall \pi_2 a_{\pi_1} \wedge \neg a_{\pi_2}$ has one trace variable $\pi_1$ that is free.
A HyperLTL formula $\phi$ is said to be closed if it does not have free variables.

Observe that one may compile a HyperLTL formula $\psi$ with $p$-trace variables and no trace quantification to an NBA \cite{clarkson_temporal_2014,finkbeiner2023automata} $\Aut^p = (\Sigma', Q, q_0, \delta, Q_{Acc})$, where the alphabet $\Sigma' = (2^\AP)^{p}$ consists of $p$-tuples of sets of atomic propositions.
For convenience, we refer to formulae of the form $\psi$ as quantifier-free formulae in HyperLTL.
While LTL formulae are defined with respect to traces, HyperLTL formulae are defined with respect to a set of traces $T \subseteq \Sigma^{\omega}$.
We refer the reader to \cite{clarkson_temporal_2014} for more details on the syntax and semantics of HyperLTL.
We use  $\system \models_{\lab} \phi$ to denote that the system under a labeling $\lab$ satisfies a closed HyperLTL formula $\phi$.
Throughout this paper, we are concerned with  HyperLTL formulae of the form $\quant_1 \pi_1 \quant_2 \pi_2, \ldots \quant_p \pi_p \psi$.
where $\quant$ is used as syntactic sugar for either the quantifier $\forall$ or $\exists$ and $\psi$ is a quantifier-free HyperLTL formula. Without loss of generality, we assume that $\quant_1 = \forall$ for notational convenience.
A property is in the fragment $\forall^* \exists^*$ if it is of the form $\phi = \forall \pi_1 , \ldots, \forall \pi_v \exists \pi_{v+1}, \ldots, \exists \pi_p \psi$, where $\psi$ is a quantifier-free HyperLTL formula.

\noindent \textbf{Game-based semantics.}
The satisfaction of a HyperLTL formula can be thought of as a game between two players, the universal and existential player, respectively.
The game proceeds in turns as follows. The player corresponding to the first quantifier $\quant_1$ selects a trace for $\pi_1$, and the game proceeds to the next player. 
Then the player controlling the quantifer $\quant_2$ selects a trace $\pi_2$.
This continues until the last player selects the trace $\pi_p$.
If the selected set of traces ensures the satisfaction of 
formula $\psi$, then formula $\phi$ holds, and the existential player wins.
This corresponds to the case where the system satisfies the desired specification, \textit{i.e.}, $\system \models_{\lab} \phi$.
If the selected traces do not satisfy $\psi$ then the universal player wins and the system does not satisfy the formula.

\noindent \textbf{Illustrative Example.}
To illustrate such properties and their use, consider the set of atomic propositions $\AP = \{s, o \}$ where $s$ indicates a state is secret and the presence or absence of $o$ indicates a binary output that is observable from a state.
Then opacity \cite{Bryan_2006_opacity} can be characterized as the HyperLTL formulae $\forall \pi_1 \exists \pi_2 s_{\pi_1} \implies \Big( \neg s_{\pi_2} \wedge \always(o_{\pi_1} \Leftrightarrow o_{\pi_2}) \Big)$.
That is, given any trace of the system $\pi_1$ that starts from some secret state, one may always find a trace $\pi_2$ that starts from a non-secret state and ensures that the outputs of the traces are always the same for all time.
We can compile the quantifier-free part of the opacity in the automaton in Figure \ref{subfig:aut_opacity} and its negation in Figure \ref{subfig:aut_opacity_comp}.
Here, we use logical formulae over $\AP$ to indicate the transitions in the automaton for convenience.
For example, the logical formula $s_{\pi_1} \wedge \neg s_{\pi_2} \wedge (o_{\pi_1} \Leftrightarrow o_{\pi_2})$ on the edge between states $q_0$ and $q_1$ is satisfied by the pair of sets $( \{s, o\}, \{o \})$, and the pair $( \{ s\}, \{ \})$. 
Here, the first pair indicates that the first trace $\pi_1$ has a secret state, the second $\pi_2$ does not, and the outputs of both states are $o$.
Similarly, the second pair indicates that the first trace has a secret, the second a non secret and that they both have no outputs.
The automaton then transitions from one state $q$ to another $q'$ if it reads any pair of sets that satisfy the logical formula.
In the above example, either pair can cause the NBA in Figure \ref{subfig:aut_opacity_comp} to transition from the initial state $q_0$ to an accepting state $q_1$.
\begin{figure}
    \centering
\begin{subfigure}{0.5\textwidth}
\centering
\begin{tikzpicture}[node distance=2.5cm,thick]
     \node[initial, accepting, state, initial text=,] (0) {$q_0$};
     \node[state, below right of = 0,, accepting ] (1) {$q_1$};
     \node[state, above right of = 1, ] (2) {$q_2$};
     \path [->]
     (0) edge node[above]{$s_{\pi_1} \wedge s_{\pi_2}$} (2)
     (0) edge  node[left] {$s_{\pi_1} \wedge \neg s_{\pi_2} \wedge (o_{\pi_1} \Leftrightarrow o_{\pi_2})$} (1) 
     (1) edge node[right]{$\neg (o_{\pi_1} \Leftrightarrow o_{\pi_2})$} (2) 
     (1) edge[loop below] node{$(o_{\pi_1} \Leftrightarrow o_{\pi_2})$} (1)
     (2) edge [loop above] node{$(\top, \top)$} (2);
     \end{tikzpicture}
\caption{Automaton representing the quantifier-free part of the opacity specification. This automaton accepts a pair of traces if they are opaque, that is, they have the same sequence of outputs. Here, state $q_1$ is accepting.}
\label{subfig:aut_opacity}
\end{subfigure}
\begin{subfigure}{0.5\textwidth}
\centering
\begin{tikzpicture}[node distance=2.5cm,thick]
     \node[initial, state, initial text=, ] (0) {$q_0$};
     \node[state, below right of = 0, ] (1) {$q_1$};
     \node[state, accepting, above right of = 1, ] (2) {$q_2$};
     \path [->]
     (0) edge node[above]{$s_{\pi_1} \wedge s_{\pi_2}$} (2)
     (0) edge  node[left] {$s_{\pi_1} \wedge \neg s_{\pi_2} \wedge (o_{\pi_1} \Leftrightarrow o_{\pi_2})$} (1) 
     (1) edge node[right]{$\neg (o_{\pi_1} \Leftrightarrow o_{\pi_2})$} (2) 
     (1) edge[loop below] node{$(o_{\pi_1} \Leftrightarrow o_{\pi_2})$} (1)
     (2) edge [loop above] node{$(\top, \top)$} (2);
     \end{tikzpicture}
\caption{Automaton representing the complement of opacity. This automaton accepts a pair of traces if they falsify opacity. Here, state $q_2$ is accepting.}
\label{subfig:aut_opacity_comp}
\end{subfigure}
\end{figure}

We should note that some of our techniques are also applicable to other hyperproperties that may not be captured via HyperLTL but are still expressible by a $\forall^* \exists^*$ quantification, such as $\delta$-approximate opacity \cite{zhang_2019_opacity}. 
\section{Problem Statement}
\label{sec:prob_statement}
We present a sound verification method to determine whether a given system satisfies a HyperLTL formula by using closure certificates \cite{murali2024closure}.
To do so, we draw inspiration from \cite{liu2020verification,anand2024verification}, which used barrier certificates on a self-composed system to verify HyperLTL formulae including opacity. 
We should note that in general the problem of verifying discrete-time systems against the HyperLTL formula is undecidable, as even simple safety properties are undecidable \cite{chatterjee2020polynomial}.
\begin{problem}
    Given a finite set of atomic propositions $\AP$, a system $\system = (\states, \initstates,\transfunc)$, a labeling map $\lab: \states \to (2^{\AP})$, and a  HyperLTL formula $\phi$, determine whether $\system \models_{\lab} \phi$.  
\end{problem}

We now demonstrate why augmented barrier certificates may fail to prove opacity with the help of finite-state examples, and then demonstrate how our approach allows for one to prove the opacity for such systems.

\section{Augmented Barrier Certificates for Verifying HyperLTL formula}
\label{sec:abc_hyperLTL}

We briefly explain the concept of augmented barrier certificates for verifying HyperLTL formulae, as used in \cite{anand2024verification}.
Intuitively, the goal of augmented barrier certificates is to find an appropriate safety argument in the $p$-fold composition of the system with itself to ensure that a system starting from the initial augmented state does not reach an accepting state.
We now illustrate its limitations by considering the hyperproperty of opacity.
We should note that the definition of augmented barrier certificates we consider here is more general than those present in \cite{liu2020verification,anand2024verification} as we consider quantifier alternations for the initial set of states.

\noindent \textbf{Illustrative Example.}
Consider the illustrative example of opacity and the NBA that denotes the complement in Figure \ref{subfig:aut_opacity_comp}.
Observe that if a pair of traces reach the accepting state $q_2$, then there is no way for the existential trace to exit this accepting state.
Thus, the authors of \cite{anand2024verification} seek to find appropriate augmented barrier certificates $\barrier : \states \times \states \to \Reals$ that seek to separate the initial state from an accepting state to guarantee opacity.
For ease of explanation, we use the following notation.
Let $\states_s = \{ x \mid \lab(x) = s \}$ denote the secret set of states and $\states_{ns}$ the non-secret states.
For any two states $x,y\in \states$, let $x \simeq_{o} y$ indicate that they agree on the output.
That is, if $x \simeq_{o} y$, then $o \in \lab(x)$ and $o \in \lab(y) $, or $o \notin \lab(x)$, and $o \notin \lab(y)$.
We use  $x \not \simeq_{o} y$ when the states $x,y$ do not agree on the output.

We define augmented certificates for opacity as follows:
\begin{definition}
\label{def:abc}
Consider a system $\system = (\states, \initstates, f)$, the set of atomic propositions $\AP = \{s, o\}$, and a labeling map $\lab: \states \to \Sigma$.
A function $\barrier: \states \times \states \to \Reals$ is an augmented barrier certificate if for all $x_s \in \initstates  \cap \secretstates$, there exists a state $x_{ns} \in \initstates \setminus \secretstates$, such that $x_s \simeq_{o} x_{ns} $, and :
\begin{align}
& \barrier(x_s, x_{ns}) \geq 0, \label{eq:abc_init}
\end{align}
and for all states $x,y \in \states$, such that  $x_s \not \simeq_{o} x_{ns} $, there exists $\xi \in \Reals_{ > 0}$ such that:
\begin{align}
& \barrier(x,y) \leq -\xi, \label{eq:abc_unsafe}
\end{align}
and for all states $x,y \in \states$, and for all $x' \in f(x)$, there exists $y' \in f(y)$ such that:
\begin{align}
     \big( \barrier(x,y) \geq 0 \big) \implies \big( \barrier(x',y') \geq 0 \big). \label{eq:abc_inv}
\end{align}
\end{definition}

Intuitively, the existence of an augmented barrier certificate as in Definition \ref{def:abc} ensures that a system satisfies the opacity by ensuring that a system does not reach the accepting state in the NBA in Figure \ref{subfig:aut_opacity_comp}.
Condition~\eqref{eq:abc_init} ensures that for any secret state there exists a similar non-secret state (i.e., with the same output), where the barrier value is nonnegative.
This corresponds to a transition from state $q_0$ to $q_1$ in Figure \ref{subfig:aut_opacity_comp}.
 Condition~\eqref{eq:abc_unsafe} ensures that when two states have different outputs, their barrier value must be strictly negative. 
This corresponds to states that cause the NBA to transition from state $q_1$ to $q_2$.
Finally, condition~\eqref{eq:abc_inv} guarantees that the barrier certificate does not become negative along transitions and thus cannot reach a state with a negative barrier value.
Hence condition~\eqref{eq:abc_inv} ensures that a system that reaches state $q_1$ continues to loop on the state $q_1$ and thus acts as a ``barrier'' that prevents the system from reaching state $q_2$.

Now consider the finite system in Figure~\ref{fig:abc_fails} to illustrate a case where an augmented barrier certificate cannot be found, as in Definition \ref{def:abc} although the system is opaque. 
Formally, consider a system $\system = (\states,\initstates,f)$ as in Figure~\ref{fig:abc_fails} with eleven states: $\states = \{x_0, x_1, \ldots, x_{10} \}$, and the initial set of states $\initstates = \{ x_0, x_5, x_8\}$.
Here, $x_0$ is designated as a secret state and has a label $s$. The states are color-coded according to their outputs. States that have an output $o$ are colored blue, and those that do not are colored red.

The system in Figure~\ref{fig:abc_fails} is opaque because for each state sequence originating from the secret state $x_0$, there exists a corresponding sequence beginning from a non-secret state where the outputs are the same.
Specifically, there are two state sequences starting from $x_0$: $\langle x_0, x_1, x_2, \ldots \rangle$ and $\langle x_0, x_3, x_4, \ldots \rangle$. 
Similarly to these, we find sequences beginning from non-secret states: $\langle x_5, x_6, x_7, \ldots \rangle$ and $\langle x_8, x_9, x_{10}, \ldots \rangle$ that have the same sequences of the output value $o$.
Now, we illustrate how we fail to find augmented barrier certificates as in Definition \ref{def:abc}.
Following condition~\eqref{eq:abc_init}, the value of the augmented barrier certificate for any secret and some non-secret states with the same output must be nonnegative.
Thus, we have $\barrier(x_0, x_5) \geq 0$, or $\barrier(x_0, x_8) \geq 0$.
Suppose that we have $\barrier(x_0, x_5) \geq 0$.
Then following condition~\eqref{eq:abc_inv}, we require that the value of the barrier certificate not decrease along transitions, and so we must have:
\begin{itemize}
    \item $\barrier(x_1, x_6) \geq 0$ and $\barrier(x_1, x_9) \geq 0$, and, 
    \item $\barrier(x_2, x_7) \geq 0$, and $\barrier(x_2, x_{10}) \geq 0$.
\end{itemize}
However, condition~\eqref{eq:abc_inv} contradicts this by requiring that if two states have different outputs, then the barrier value for that pair must be strictly negative. 
Thus, we must have $\barrier(x_2, x_{10}) < 0$ which is a contradiction.
In a symmetric fashion, we obtain a contradiction with the value of $\barrier(x_4, x_7)$, so we can conclude that there is no augmented barrier certificate as in Definition \ref{def:abc} to ensure the opacity of the above system.

The key reason for why augmented barrier certificates fail is due to a lack of \emph{lookahead}.
In particular, the choice of the trace $\pi_2$ can depend on the future actions of trace $\pi_1$.
If the existential player could know that the state $x_2$ or $x_4$ would eventually be reached by the trace $\pi_1$, then the trace $\pi_2$ can start from $x_5$ or $x_8$. 
However, augmented barrier certificates do not capture this ability of the existential player to look at the future choices of the universal player.

We now consider the notion of closure certificates to capture this lookahead.
Consider a function $\transclos: \states \times \states \to \Reals$, that satisfies conditions \eqref{eq:cc_base} and \eqref{eq:cc_inv}. 
Then for any state $y$ that can be reached from an initial state $x$, the value of the function $\transclos(x,y)$ is nonnegative.
We formalize and prove this result in the following lemma.

\begin{lemma} \label{lemma:transclos}
Consider a system $ \system = (\states, \initstates, \secretstates, \transfunc) $, a closure certificate $ \transclos: \states \times \states \rightarrow \Realpos $, and a state $x \in \states$.
Then for any $i \in \Nats_{\geq 1}$, we have $\transclos(x,f^i(x)) \geq 0$.
\end{lemma} 

\begin{proof}
    Consider a state $y = f^i(x) \in \states $ for some $x \in \states$ and $i \in \Nats_{\geq 1}$. 
    Following condition \eqref{eq:cc_base}, we have $ \transclos(f(x), x) \geq 0 $, and similarly, $\transclos(f^j(x), f^{j+1}(x)) \geq 0$ for all $j \in \{0, \dots, i\}$.
 By condition \eqref{eq:cc_inv} and induction, we have $ \transclos(x, f^j(x)) \geq 0 $ for all $ i \in \{1, \dots, i\} $ and thus $ \transclos(x, f^i(x)) \geq 0 $.
\end{proof}
Intuitively, the function $\transclos$ acts as a lookahead. 
For any states $x, y \in \states$, we can conclude that $y$ can be reached from $x$ if $\transclos(x,y) \geq 0$ and so we can define functions with respect to this condition.
For convenience, we define the following primitive formula: 
\begin{align*}
&\InPath(x, y, z) = \big(\transclos(x, y) \geq 0 \big) \land \big( \transclos(y, z) \geq 0 \big)  \\
& \qquad \qquad \qquad \qquad \lor (y = z),
\end{align*}
where $ x $, $ y, z \in \states$ are states of the system. Intuitively, $ \InPath(x, y, z) $ is true if there is a path from state $ x $ to state $ z $ and state $ y $ is present in the path.
to capture this notion of lookahead.
We formalize this result in the following lemma:
\begin{lemma}\label{lemma:InPath}
Consider a system $\system = (\states, \initstates, \secretstates, \transfunc)$ and a state $x_0 \in \states$.
 If there exists a state sequence $ \langle x_0, x_1, \ldots, x_n, \ldots \rangle $, then $ \InPath(x_0, x_i, x_n) $ is true for all $ i \in \{0, \ldots, n\} $.
\end{lemma}

\begin{proof}
We consider two cases:
\begin{itemize}
    \item [(1)] If $ x_i = x_0 $ or $ x_i = x_n $, the second or third disjunction in $ \InPath(x_0, x_i, x_n) $ is satisfied, and thus the formula is true.
    \item [(2)] Otherwise, $ x_i $ is an intermediate state on the path from $ x_0 $ to $ x_n $, and by Lemma~\ref{lemma:transclos}, we have $ \transclos(x_0, x_i) \geq 0 $ and $ \transclos(x_i, x_n) \geq 0 $. Therefore, $ \InPath(x_0, x_i, x_n) $ is true.
\end{itemize}
\end{proof}

We now show how one may use the conditions over closure certificates (as well as the $\InPath$ formula) to act as a lookahead for the verification of HyperLTL formulae.
For ease of exposition, we first present a definition of HyperCertificates specifically for opacity before we generalize these for other hyperproperties.

\section{HyperCertificates for opacity}
\label{subsec:ac_opacity_1}
To define a notion of HyperCertificates for opacity, we observe that if the traces of the system reach the accepting state of the NBA in Figure \ref{subfig:aut_opacity_comp}, then they continue to remain there regardless of a choice made by the players.
Thus, we can consider a notion of  HyperCertificates with lookahead to ensure safety.
\begin{definition}
    \label{def:acc_opac_1}
    Consider a system $\system = (\states, \initstates, f)$, the set of atomic propositions $\AP = \{s, o \}$, and a labeling map $\lab: \states \to \Sigma$. 
   A HyperCertificate for opacity is a pair  of functions $(\transclos, \barrier)$ where $\transclos: \states \times \states \to \Reals$ is a closure certificate as in Definition \ref{def:cc}, and $\barrier: \states \times \states \times \states  \to \Reals$ satisfies the following conditions. 
   There exists $\xi \in \Reals_{ > 0}$, such that for all states $x_0 \in \states_s$, and any state $z \in \states$, where $\transclos(x_0,z) \geq 0$, there exists $x_{ns} \in \states_{ns}$, such that  $x_0 \simeq_{o} x_{ns}$ and:
    \begin{align}
        & \barrier(x_0, x_{ns}, z) \geq 0, \text{ and }
        \label{eq:acc_opaque_bar_base}
    \end{align}
    for all states $x, y \in \states$ where $x \not \simeq_{o} y$, we have 
    \begin{align}
        & \barrier(x, y,z) \leq -\xi, \text{and} \label{eq:acc_opaque_bar_uns}
    \end{align}
    for all states $x,y \in \states$, and for all $x' \in f(x)$, where $\InPath(x_0, x', z)$ is true, there exists $y' \in f(y)$ such that
    \begin{align}
        & \big( \barrier(x,y,z) \geq 0 \big) \implies \big( \barrier(x',y',z) \geq 0 \big). \label{eq:acc_opaque_bar_next}
    \end{align}
\end{definition}

We provide an intuitive explanation of our approach. 
Recall from Lemma \ref{lemma:transclos} that if $\transclos(x_0,z) \geq 0$ then there may exist a path from a secret state (denoted by $x_0$) to state $z$. 
Our aim is to show the existence of a path from a non-secret state $x_{ns}$ with the same observations given the lookahead information of state $z$. 

The function $\barrier$ is defined so that its first two inputs represent the states of paths originating from an initial secret and non-secret states, respectively, while the third input represents the lookahead state of the path from the secret state.
Condition~\eqref{eq:acc_opaque_bar_base} requires the existence of a non-secret state $x_{ns}$ that is observationally indistinguishable from $x_0$, representing the initial non-secret state from which an opaque, indistinguishable path can be constructed. 
Condition~\eqref{eq:acc_opaque_bar_uns} ensures that the barrier certificate is negative over pairs of states whose output observations differ, this condition similar to the conditions in the augmented barrier approach.
Finally, condition~\eqref{eq:acc_opaque_bar_next} ensures that for every successor state on the path from the first argument state in the barrier function, there exists a corresponding successor state on the second argument path with a barrier value of at least zero. This condition guarantees that for every path from $x_0$ to $z$, there exists a corresponding path from $x_{ns}$ such that they remain indistinguishable based on the output.

We now prove the soundness of these conditions in the next theorem.
\begin{theorem}\label{thm:opaue}
    Consider a system $\system$, the set of atomic propositions $\AP = \{s, o \}$, and a labeling map $\lab: \states \to \Sigma$ . The existence of a HyperCertificate $(\transclos, \barrier)$ as in Definition \ref{def:acc_opac_1} guarantee that the system $\system$ is opaque. 
\end{theorem}

\begin{proof}
    We prove this claim by contradiction.
    Consider an arbitrary length state sequence of length $n$ for any $n \in \Nats$ from a secret state $x_0$, $\rho = \langle x_0, x_1, \ldots, x_n  \rangle$, where $s \in \lab(x_0)$. 
    Assume, in contradiction, that for all state sequences $\rho' = x'_0, \ldots x'_n $ from a non-secret state ($s \notin \lab(x'_0)$)  there exists some $j \in \{0, \ldots, n \}$ such that $x_j \not \simeq_{o} x'_j$.
     Following Lemma~\ref{lemma:transclos}, we know that for each $i \in \Nats$, we have $\transclos(x_0, x_i) \geq 0$. 
     Following condition~\ref{eq:acc_opaque_bar_base} we must have some $x'_0$ such that $\barrier(x_0, x'_0, x_n) \geq 0$. 
     Consider the state sequences $\rho' = x'_0, \ldots, x'_n$ from this initial state $x'_0$.
     By assumption, there exists $j \in \{0, \ldots, n \}$, such that 
    Following condition~\ref{eq:acc_opaque_bar_next} and Lemma~\ref{lemma:InPath}, we have $\barrier(x_i, x'_i, x_n) \geq 0$ for all $i \in \{1, \ldots, n\}$. 
    However, condition~\ref{eq:acc_opaque_bar_uns} requires $\barrier(x_j, x'_j, x_n) \leq -\xi < 0$, which is a contradiction.
\end{proof}

Recall Figure~\ref{fig:abc_fails} to examine how the approach outlined above successfully provides a valid certificate to verify opacity in this example. 
By Lemma~\ref{lemma:transclos}, we know that $\transclos(x_0, x_2) \geq 0$ and $\transclos(x_0, x_4) \geq 0$, which confirms that states $x_2$ and $x_4$ are reachable from the secret state $x_0$ for any function $\transclos$ satisfying conditions \eqref{eq:cc_base} and \eqref{eq:cc_inv}.
An example of such a function is to define $\transclos(x,y) \geq 0$, if and only if the state $y$ is reachable from the state $x$.
Similarly, we have $\transclos(x_0, x_1) \geq 0$ and $\transclos(x_0, x_3) \geq 0$.
Next, applying condition~\ref{eq:acc_opaque_bar_next}, we select non-secret states $x_5$ and $x_8$ establishing that $\barrier(x_0, x_5, x_3) \geq 0$ and $\barrier(x_0, x_8, x_4) \geq 0$. This ensures the existence of non-secret initial states with paths that match the observations from $x_0$ to $x_1$ and $x_2$.
Finally, using condition~\ref{eq:cc_inv}, we verify that for each transition along the path from $x_0$, there exists a matching transition in the paths from the non-secret states $q_5$ and $q_8$. Specifically, we get $\barrier(x_1, x_6, x_2) \geq 0$ and $\barrier(x_2, x_7, x_2) \geq 0$.
Similarly, we get $\barrier(x_3, x_9, x_4) \geq 0$ and $\barrier(x_4, x_{10}, x_{10}) \geq 0$.
For the cases where $z = x_1$, or $z = x_3$, we may choose $x_5$, or $x_8$ for both. 
This combination of barrier and closure certificates successfully verifies that the system is opaque.

In the following section, we describe a generalization of the above to provide automata based certificate construction for $\forall^* \exists^*$ HyperLTL formulae and then show how to generalize these to general HyperLTL formulae.

\section{$\forall^* \exists^*$ HyperCertificates}
\label{subsec:forall_exist_certs}
For convenience, we first consider the HyperLTL formulae in this Section to be in the $\forall^* \exists^*$-fragment, \textit{i.e.}, of the form  $\phi = \forall \pi_1 , \ldots, \forall \pi_v \exists \pi_{v+1}, \ldots, \exists \pi_p \psi$, where $\psi$ is a quantifier-free HyperLTL formula and let $b = p - v$.
To verify our system against a hyperproperty in the $\forall^* \exists^*$-fragment, we first take a composition of the system with itself.
Formally, consider the product system as $\prodsys= (\prodstate, \prodstate_0, \prodtrans)$, where $\prodstate = \states^p$ denotes the set of states, $\prodinitstate = \initstates^p$ the set of initial states, and $(x'_1, \ldots, x'_p) \in \prodtrans \big(  (x_1, x_2, \ldots, x_p) \big)$  if $x'_i \in f(x_i)$ for all $1 \leq i \leq p$.
For convenience, we use $\bm{x} = (x_1, x_2, \ldots, x_p)$ to denote a state of the product system.
Observe that in the product system, the universal player is allowed control over the first $v$ components, while the existential player controls the last $b$ components.
Thus, when we select a successor state $\bm{x}'$, we let the universal player select the first $v$ components whereas the existential player controls the last $b$ components. 
Given a HyperLTL formula $\phi = \forall \pi_1, \ldots \forall \pi_v \exists {\pi_{v+1}} \ldots, \exists \pi_p \psi$, under a game-based semantics \cite{finkbeiner2021model}, we can consider the verification problem as determining whether the traces of the product system visit the accepting state of the NBA corresponding to $\psi$ infinitely often.
If so, we can conclude that $\system \models_{\lab} \phi$.

We now present our set of conditions to verify the systems against $\forall^* \exists^*$ HyperLTL formulae.
Similarly to opacity, our conditions rely on closure certificates \cite{murali2024closure} for the lookahead.
However, since we are dealing with general LTL properties (and no longer safety), we instead consider a combination of closure certificates with B\"uchi ranking functions as in \cite{chatterjee2024sound} as follows.
\begin{definition}
    \label{def:Hyper_cert_FE}
    Consider a system $\system = (\states, \initstates, f)$, a finite set of atomic propositions $\AP$, a labeling map $\lab$, and a HyperLTL formula $\phi = \forall \pi_1, \ldots \forall \pi_p \exists {\pi_{v+1}} \ldots, \exists \pi_p; \psi$.
    Consider the product system $\prodsys$, and let NBA $\Aut = (\Sigma, Q, q_0, \delta, Q_{Acc})$ represent  the LTL-specification $\psi$.
    Let  $\xi \in \Reals_{ > 0}$ be a non-negative value, then a $\forall^* \exists^*$ HyperCertificate is a pair of functions $(\transclos, \barrier)$, where  $\transclos: \states \times \states \to \Reals$ is a closure certificate satisfying conditions \eqref{eq:cc_base} and \eqref{eq:cc_inv}, and  function $\barrier: \prodstate \times \states^v \times Q\to \Reals_{ \geq 0}$ that satisfies the following conditions, where for all states $x_{0,1}, x_{0,2}, \ldots, x_{0,v} \in \initstates$, and for all states $\bm{z} = (z_{1}, z_2, \ldots z_v) \in \states^v$, where $\transclos(x_{0,i}, z_i) \geq 0$, there exist states $x_{0,v+1}, \ldots, x_{0,p}$ such that: 
\begin{align}
\barrier( \bm{x}_0, \bm{z}, q_0) \geq 0, \label{eq:Hyper_Cert_FE_base}
\end{align}
where $\bm{x}_0 = (x_{0,1}, \dots, x_{0,p}) \in \prodstate_{0}$.
And for all states $\bm{x} = (x_1, \ldots, x_p) \in \prodstate$, and for all states $x'_i \in f(x'_i)$ for all $1 \leq i \leq v$ where $\InPath(x_i, x'_i, z_i)$, there exist states $x'_{j} \in f(x_{j}) $, for all $v+1 \leq j \leq p$, such that:
\begin{align}
   & \big(\barrier(\bm{x},\bm{z},q) \geq 0 \big)  \implies \nonumber \\ &  \big(\barrier(\bm{x}', \bm{z}, q') \geq 0 \big), \text{ for all } q \in Q \cap Q_{Acc}, \text{ and } \label{eq:Hyper_Cert_FE_buchi_nonincrease} \\
   &\big(\barrier(\bm{x},\bm{z},q) \geq 0 \big) \implies \big(0 \leq \barrier( \bm{x}', \bm{z}, q') \big) \wedge \nonumber \\ &\big(\barrier( \bm{x}', \bm{z}, q') \leq \barrier(\bm{x}, \bm{z}, q) - \xi \big), \text{ for all } q \in Q \setminus Q_{Acc} \label{eq:Hyper_CERT_FE_buchi_rank}
\end{align}
where $\bm{x}' = ( x'_1, x'_2, \ldots x'_p)$, and $q' \in \delta(q, (\lab(x_0), \ldots, \lab(x_p))$.
\end{definition}

We now show how HyperCertificates can be used to verify systems against HyperLTL formulae in the $\forall^* \exists^*$-fragment.

\begin{theorem}
    \label{thm:acc_hyper_verif}
    Consider a system $\system = (\states, \initstates, f)$ and a $\forall^* \exists^*$-HyperLTL formula $ \phi = \forall \pi_1 , \ldots, \forall \pi_p \exists \pi_{v+1}, \ldots, \exists \pi_p \psi$.
    Let $\Aut = (\Sigma, Q, q_0, \delta, Q_{Acc})$ correspond to the quantifier-free formulae $\psi$.
 Then the existence of a HyperCertificate as in Definition \ref{def:Hyper_cert_FE} ensures that $\system \models_{\lab} \phi$. 
\end{theorem}

\begin{proof}
    We prove this by contradiction.
 Suppose that there exists a choice of traces for the universal player such that no matter the choice of the existential player, the resulting trace is not accepted.
    Formally, let the universal player select the state sequences for the first $v$ components starting from the initial states $(x_{0,1}, \ldots, x_{0,v})$  and let the existential player select the state sequences for the last $b$ components from the initial states $(x_{0,{v+1}}, \ldots, x_{0,v+b})$.        
    Then $\bm{x} = (x_{0,1}, \ldots, x_{0,v},x_{0,{v+1}}, \ldots, x_{0,v+ b} ) \in \bm{\initstates}$ represents the initial state.
    Let these state sequences correspond to the traces $\langle \lab(x_{0,i}), \lab(x_{1,i}), \ldots \rangle $ for all $1 \leq i \leq p$.
    Then we know that NBA $\Aut$ does not accept the infinite trace $\langle (\lab(x_{0,1}), \ldots, \lab(x_{0,p})), (\lab(x_{1,1}), \ldots, \lab(x_{1,p})), \ldots  \rangle$ and thus visits the accepting states $Q_{Acc}$ only finitely.
    Let the index where the product system stops visiting $Q_{Acc}$ infinitely be denoted as $j$ for some $j \in \Nats$.
    For all $\bm{z} = (z_1, \ldots, z_v)$ such that $\transclos(x_{0,i}, z_i) \geq 0$, we have $\barrier(\bm{x}_0, \bm{z}, q_0) \geq 0$ as in  \eqref{eq:Hyper_Cert_FE_base}.
 Following conditions \eqref{eq:Hyper_Cert_FE_buchi_nonincrease} and \eqref{eq:Hyper_CERT_FE_buchi_rank} and through induction, we have $\barrier(\bm{x}_{(j+i)}, \bm{z}, q') \leq \barrier(\bm{x}_{j}, \bm{z}, q) - i \xi$ for all $q \in Q \setminus Q_{Acc}$ and all $i \in \Nats_{\geq 1}$.
    As $i$ goes to $\infty$ we have $\barrier(\bm{x}_{j_i}, \bm{z}, q') < 0$ which is a contradiction.
\end{proof}

Note that conditions \eqref{eq:Hyper_Cert_FE_buchi_nonincrease} and \eqref{eq:Hyper_CERT_FE_buchi_rank} require that the existential player \emph{first} selects the next state regardless of the state of the automaton (the choice of existential state $x'_j \in f(x_j)$ must be the same for all choices of $q \in Q$).
We may swap these for restricted classes of hyperproperties as discussed in \cite[Section III.D]{anand2024verification}.
While the above approach is more general than augmented barrier certificates, one may still fail to find HyperCertificates in cases where multiple distinct paths share the same initial and end states even if the system satisfies the desired specification as illustrated below.

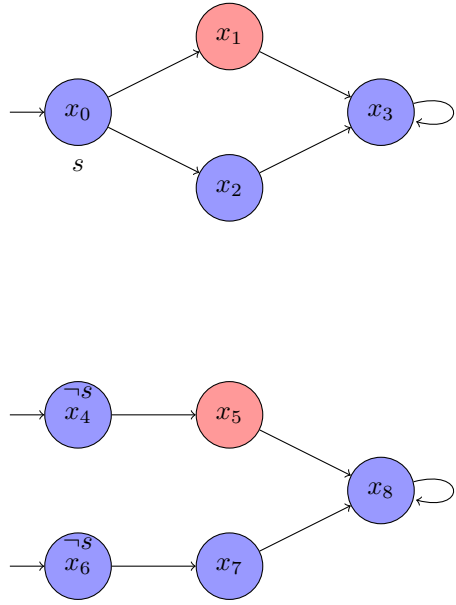
\begin{figure}[h]
   \begin{tikzpicture}
    \begin{scope}[xshift=0cm] 
    \node[initial, state, initial text=, fill=blue!40!white] at (0,0) (0) {$x_0$};
    \node[below=0.05cm of 0] {$s$};
     \node[state,fill=red!40!white ] at (2,1) (1) {$x_1$};
     \node[state,fill=blue!40!white ] at (2,-1) (2) {$x_2$};
     \node[state,fill=blue!40!white ] at (4,0) (3) {$x_3$};
     
    \node[initial,initial text=, state,fill=blue!40!white ] at (0,-4) (4) {$x_4$};
    \node[below=0.05cm of 5] {$\neg s$};     
        \node[state,fill=red!40!white ] at (2,-4) (5) {$x_5$};

     \node[initial,initial text=, state,fill=blue!40!white ] at (0,-6) (6) {$x_6$};
     \node[below=0.05cm of 8] {$\neg s$};
     \node[state,fill=blue!40!white ] at (2,-6) (7) {$x_7$};
     \node[state,fill=blue!40!white ] at (4,-5) (8) {$x_{8}$};

     \path[->]
     (0)  edge node{} (1)
     (0)  edge node{} (2)
     (1)  edge node{} (3)
     (2)  edge node{} (3)
     (3) edge[loop right] node{} (3)
     
     (4)  edge node{} (5)
     (5)  edge node{} (8)
     (6)  edge node{} (7)
     (7)  edge node{} (8)
     (8) edge[loop right] node{} (8);
    \end{scope}
   \end{tikzpicture}
   \caption{A finite state example illustrating the limitations of $\forall^* \exists^*$ certificates as in Definition\ref{def:acc_opac_1} indicating the limitations of our approach.}
   \label{fig:mult_path}
\end{figure}
\textbf{Illustrative Example:}
Consider a finite state system $\system = (\states, \initstates, f)$ as in Figure~\ref{fig:mult_path}, where $\states = \{x_0, \ldots, x_8 \}$ indicates the set of states, and $\initstates = \{ x_0, x_4, x_6 \}$ the initial set of states. Consider the secret initial state $x_0$ as the starting point of the path and $x_3$ as the end state.
By Lemma~\ref{lemma:transclos}, we know that $\transclos(x_0, x_3) \geq 0$, since $x_3$ is reachable from $x_0$. 

According to condition~\eqref{eq:acc_opaque_bar_base}, we must select a non-secret state to match this path.
The options are $\barrier(x_0, x_4, x_3) \geq 0$ and $\barrier(x_0, x_6, x_3) \geq 0$.
Next, condition~\eqref{eq:acc_opaque_bar_next} requires that there exist intermediate states along these paths to maintain the barrier condition. 
Specifically, we would need $\barrier(x_2, x_5, x_3) \geq 0$ or $\barrier(x_1, x_7, x_3) \geq 0$. However, according to condition~\eqref{eq:acc_opaque_bar_uns}, the pairs $(x_2, x_5)$ and $(x_1, x_7)$ do not have the same outputs; thus, we require that $\barrier(x_2, x_5, x_3) < 0$ and $\barrier(x_1, x_7, x_3) < 0$, which leads to a contradiction.
Such a limitation also persists in the case of  HyperCertificates.
An approach to eliminate this problem is to consider $k$-length sequences instead of a single-step lookahead that captures only one state similar to the results in \cite{finkbeiner2023automata}.
However, this would require parameterizing the lookahead $\bm{z}$ with a $k$-tuple of elements for each of the $v$ traces.
We adopt a similar approach as above for general hyperproperties in the following.

\section{HyperCertificates for General Hyperproperties}
\label{sec:gen}
We now present a set of conditions to verify systems against general HyperLTL formulae.
Similarly to the approach for $\forall^* \exists^*$ hyperproperties, we make use of certificates over the product of the product system and the NBA representing the quantifier-free part of the HyperLTL formula.
For notational convenience, we use $\phi = \quant_1 \pi_1, \quant_2 \pi_2, \ldots, \quant_p \pi_p \psi$ to denote arbitrary HyperLTL formulae,  where $\quant_i \in \{ \forall, \exists \}$ for all $i \in \{1, \ldots, p\}$. Without loss of generality, we assume that $\quant_1 = \forall$.
Let $v$ denote the number of universally quantified traces and $b$ denote the number of existentially quantified traces, with $p = v + b$. We define the index sets $\indexforall = \{ j_1, j_2, \ldots, j_v \} \subseteq \{1, \ldots, p\}$ for the positions of the $\forall$ quantifiers and $\indexexists = \{ \ell_1, \ell_2, \ldots, \ell_b \} = \{1, \ldots, p\} \setminus \indexforall$ for the positions of the $\exists$ quantifiers. By assumption, we set $j_1 = 1$.

\begin{definition}
    \label{def:def:Hyper_cert}
    Consider a system $\system = (\states, \initstates, f)$, a finite set of atomic propositions $\AP$, a labeling map $\lab$ and a HyperLTL formula $\phi = \quant_1 \pi_1, \quant_2 \pi_2, \ldots, \quant_p \pi_p, \psi$, the product system $\prodsys$ and the  NBA $\Aut = (\Sigma, Q, q_0, \delta, Q_{Acc})$ denoting the formula $\psi$.
    Let  $\xi \in \Reals_{ > 0}$ be a non-negative value, then a HyperCertificate is a pair of functions $(\transclos, \barrier)$, where  $\transclos: \states \times \states \to \Reals$ is a closure certificate satisfying conditions \eqref{eq:cc_base} and \eqref{eq:cc_inv}, and  function $\barrier: \prodstate \times \states^v \times Q\to \Reals_{ \geq 0}$ are a HyperCertificate if
for all $x_{0,1} \in \initstates$, and 
for all $ z_{1} \in \{z \mid \transclos(x_{0,1}, z) \geq 0 \}$,
and all $x_1' \in f(x_1)$, and
\begin{itemize}
    \item [(if)] $\quant_2 = \exists$, 
then there exists $x_{0, 2} \in \initstates$,
\item [(else)] for all $x_{0, 2} \in \initstates$, and for all $\forall z_{2} \in \{z \mid \transclos(x_{0,2}, z) \geq 0 \}$, and 
\end{itemize}
\begin{align}
&& \vdots  && \nonumber  
\end{align}
\begin{itemize}
    \item [(if)] $\quant_p = \exists$, 
then there exists $x_{0, p} \in \initstates$ such that:
\item [(else)] for all $x_{0, p} \in \initstates$, and for all $\forall z_{p} \in \{z \mid \transclos(x_{0,p}, z) \geq 0 \}$:
\end{itemize}
\begin{align}
\text{ we have } \barrier(\bm{x}_0, \bm{z}, q_0) \geq 0, \label{eq:Hyper_Cert_base}
\end{align}
and $\bm{z} =(z_{j_1}, z_{j_2}, \ldots, z_{j_v})$.
And  $\forall x_1 \in \states$, and $\quant_1 x'_1 \in f(x_1)$, $\forall x_2 \in \states$, such that $\quant_2 x'_2 \in f(x_2)$, $\ldots$, $\forall x_p \in \states$, $\quant_p x'_p \in f(x_p)$ where $\quant_i = \forall$ if $i \in \indexforall$, and $\InPath(x_i, x'_i, z_i)$ is true or $\quant_i = \exists$ if $i \in \indexexists$ is such that:
\begin{itemize}
    \item  for all $q \in Q \cap Q_{Acc}$,
    \begin{align}
   &  
   \big( \barrier( \bm{x}, \bm{z}, q) \geq 0 \big) \implies \big( \barrier( \bm{x}', \bm{z}, q') \geq 0 \big) , \label{eq:Hyper_Cert_accept} \\
    \end{align}
    \item and for all $q \in Q \setminus Q_{Acc}$, 
    \begin{align}
   & \big( \barrier(\bm{x},  \bm{z}, q) \geq 0 \big) \implies \big(0 \leq \barrier( \bm{x}', \bm{z}, q') \big) \wedge  \nonumber \\
   &\big(\barrier( \bm{x}', \bm{z}, q') \leq \barrier(\bm{x}, \bm{z}, q) - \xi \big), \label{eq:Hyper_Cert_rank}
\end{align}
\end{itemize}
where $\bm{x}' = ( x'_1, x'_2, \ldots x'_p)$, and $q' \in \delta(q, (\lab(x_0), \ldots, \lab(x_p))$.
\end{definition}

We now show how HyperCertificates are useful for the verification of hyperproperties specified in HyperLTL.
\begin{theorem}
    \label{thm:Hyper_Cert}
    Consider a system $\system = (\states, \initstates, f)$ and a HyperLTL formula $ \phi =\quant_1 \pi_1, \quant_2 \pi_2, \ldots, \quant_p \pi_p \psi$.
    Let $\Aut = (\Sigma, Q, q_0, \delta, Q_{Acc})$ correspond to the quantifier-free formula $\psi$, and the labeling map $\lab: X \to \Sigma$ map each state to a finite label.
    Then the existence of functions $\transclos$ and $\barrier$ as in Definition \ref{def:def:Hyper_cert} ensures that $\system \models_{\lab} \phi$. 
\end{theorem}
\begin{proof}
 Similarly to the proof of Theorem~\ref{thm:Hyper_Cert}, we prove this claim through contradiction. 
    Let us assume that there exists a choice of traces for the universal player such that no matter the choice of the existential player, the resulting set of traces is not accepted by the automaton representing the quantifier free part of the formula. 
    As mentioned previously, each player's choice of a trace is based on the order of quantification, i.e, if the quantifier at position $k$ is $\forall$, the universal player chooses this trace based on the previous $k-1$ choices made.
    Let the universal player's choices of starting states of the product system be $(x_{0,j_1}, x_{0,j_2} \ldots x_{0,j_v})$.
    Similarly, based on the lookahead state $\bm{z} = (z_{j_1}, \ldots, z_{j_v})$ where $\transclos(x_{0,j_i}, z_{j_i}) \geq 0$ for all $j_i \in \indexforall$, let the choice of initial states by the existential player be $(x_{0,l_1},x_{0,l_2} \ldots, x_{0,l_b})$.
    Then the initial state of the product will be $\bm{x}_0 = (x_{0,1} \ldots, x_{0,p})$ where $x_{0,i} = x_{0,j_k}$ when $i = j_k \in \indexforall$ and $x_{0,i} = x_{0,l_k}$ when $i = l_k \in \indexexists$ for all $1 \leq i \leq p$.
 Observe that the choices of the existential player are not made collectively after that of the universal player; rather, one selects the choices in the order of quantification.
    We connect the state sequences with the traces of the NBA $\Aut$ via the labeling function where the traces of a state sequence $\langle x_{0,i}, x_{1,i} \ldots \rangle$ give the trace $\langle \lab(x_{0,i}), \lab(x_{1,i}) \ldots$ for all $1 \leq i \leq p$.
 Then we know that NBA $\Aut$ does not accept the trace $\langle (\lab(x_{0,1}), \ldots, \lab(x_{0,p})), (\lab(x_{1,1}), \ldots, \lab(x_{1,p})), \ldots  \rangle$ and thus visits the accepting states $Q_{Acc}$ only finitely often. 
    Let $j \in \Nats$ be the index at which the product system stops visits to an accepting state of the NBA. 
    From condition~\eqref{eq:Hyper_Cert_base}, we know that for all $\bm{z} = (z_{j_1}, \ldots, z_{j_v})$ such that $\transclos(x_{0,j_i}, z_{j_i}) \geq 0$, we have $\barrier(\bm{x}_0, \bm{z}, q_0) \geq 0$.
 Following conditions \eqref{eq:Hyper_Cert_rank} and \eqref{eq:Hyper_Cert_accept} and via induction, we have $\barrier(\bm{x}_{(j+i)}, \bm{z}, q') \leq \barrier(\bm{x}_{j}, \bm{z}, q) - i \xi$ for all $q \in Q \setminus Q_{Acc}$ and all $i \in \Nats_{\geq 1}$.
    As $i$ goes to $\infty$ we have $\barrier(\bm{x}_{j_i}, \bm{z}, q') < 0$ which is a contradiction.
\end{proof}

We should note that in the above Definition, the conditions depend on the order of quantifier alternation.
However, given a HyperLTL formula, where $\quant_i$ are fixed, one can write specific conditions for HyperCertificates to verify formulae in this form.
Furthermore, we should remark that we do not allow the existential player to select their choice of the next state based on the entire current state of the product system. Doing so would allow information to ``leak'' from the choices made by the universal player in traces that succeed the existential player in the game-based semantics and would thus be invalid. 
We now discuss how one may automate the search for these certificates via approaches similar to those of barrier certificates and closure certificates.

\section{Synthesis of HyperCertificates}
\label{sec:computation}
We now present how one may employ existing automated approaches to synthesize HyperCertificates.
Specifically, we explore a Sum-of-Squares (SOS) \cite{Parrilo_2003} based formulation and a data-driven approach based on Linear Programming which we validate via satisfiability modulo theory solvers (SMT).

\subsection{Sum-of-Squares (SOS) Formulation}
\label{subsec:sos_comp}
For ease of exposition, we present how one can formulate the search for HyperCertificates  for the $\forall^* \exists^*$-fragment of HyperLTL. 
This approach can be readily extended to the general conditions presented in Section \ref{sec:gen}.
Let the HyperLTL formula be of the form $\phi = \forall \pi_1, \ldots, \forall \pi_v \exists \pi_{v+1} \ldots \exists \pi_p \psi$, and let NBA $\Aut = (\Sigma,Q, q_0, \delta, Q_{Acc})$ accept $\psi$.
A set $\mathcal{Z}\subseteq\mathbb{R}^n$ is semi-algebraic if there exists a polynomial vector $g_\mathcal{Z}$ such that $\mathcal{Z}=\{z\in\mathbb{R}^n\mid g_\mathcal{Z}(z) \geq 0\}$.
To synthesize certificates using an SOS approach, we assume that the state set $\states$, the set of initial states $\initstates$, and the relevant sets $\states_{\sigma} = \{ x \in \states \mid \lab(x) = \sigma \}$  are semi-algebraic sets \cite{bochnak2013real} for all $\sigma \in \Sigma$ defined with polynomial vectors $g(x)$, $g_0(x)$ and $g_{\sigma}(x)$, respectively.

Under this assumption, the sets $\states \times \states$, $\prodstate$, $\states^v$ and $\prodinitstate$ of the corresponding product system $\prodsys$ are also semi-algebraic.
Let these sets be defined with polynomial vectors $g_2$, $\bm{g}$, $g_v$, and $\bm{g}$, respectively. 
Finally, consider the set $\states^v \times \states^v$ to be defined with the help of the vector of polynomials $g_{2v}$. 
Next, we assume that the transition function is a polynomial in the state variable $x$, parameterize the functions $\transclos(x,y) = \sum_{l=0}^{l =z_1}t_l,\chi_l(x,y)$, and consider piece-wise functions $\barrier_q(\bm{x}, \bm{z}) = \sum_{l_q=0}^{l_q = z_2}c_{l,q}\upsilon_{q,l}(\bm{x}, \bm{z})$ for each automaton state $q \in Q$ where $t_l,c_{l,q}\in\mathbb{R}$ are unknown coefficients and $\chi_l: \states \times \states$ and $\upsilon_q: \prodstate \times \states^v $ are monomials. 

To illustrate how we handle implications, consider condition \eqref{eq:Hyper_Cert_FE_buchi_nonincrease}.
First, replace the implication $(\barrier(\bm{x}, \bm{z}, q ) \geq 0) \implies (\barrier( \bm{x}', \bm{z}, q') \geq 0 )$ with the sufficient condition $\barrier(\bm{x}', \bm{z}) \geq  \tau_0 \barrier(\bm{x}, \bm{z}  ) $ for some nonnegative bilinear value $\tau_0 \in \Reals_{\geq 0}$.
Now condition \eqref{eq:Hyper_Cert_FE_buchi_nonincrease} corresponds to the following implication:
\begin{align*}
&\big(g_v(\bm{x}_v) \geq 0 \big) \wedge \big( g_v(\bm{z}) \geq 0 \big) \wedge \big( \underset{1 \leq i \leq v}{\bigwedge} (x'_i \in f(x_i) ) \big) \wedge \\
& \underset{1 \leq i \leq v}{\bigwedge} \big(\transclos(x_i, x'_i) \geq 0) \wedge ( \transclos (x'_i, z_i) \geq 0 ) \\
& \qquad \vee (x'_i = z_i) \big) \implies \\
& \exists x_{v+1}, \ldots x_{p} \big( \underset{ \substack{q \in Q \cap Q_{Acc} \\ q' \in \delta(q, \lab(\bm{x}))}}{\bigwedge} ( \barrier(\bm{x}', \bm{z}) \geq  \tau_0 \barrier(\bm{x}, \bm{z}  ) \big),    
\end{align*}
where where $\bm{x_v} = (x_1, \ldots, x_v)$ denotes the set of states selected by the universal player, and $\bm{x} = (x_1, \ldots, x_p) \in \prodstate$ denotes the set of states selected by the players.
Observe, that the above implications consist of conjunctions and disjunctions of polynomial inequalities, (for the equalities such as $x'_i = z_i$, we can rewrite them as inequalities $z_i - x'_i \geq 0$, and $x'_i - z_i \geq 0$) except for the existential quantification of $x_{v+1} \ldots x_p$ and the nondeterministic choice for $x'_i \in f(x_i)$.
We now convert these existential quantifiers into universal quantifiers similar to \cite{jagtap_2019_formal} as follows:
\begin{align*}
&\big(\bm{g}_v(\bm{x}_v) \geq 0 \big) \wedge \big( g_v(\bm{z}) \geq 0 \big) \wedge \big( \underset{1 \leq i \leq v}{\bigwedge} (x'_i \in f(x_i) ) \big) \wedge \\
& \big( \underset{1 \leq i \leq v}{\bigwedge}(  (\transclos(x_i, x'_i) \geq 0) \wedge ( \transclos (x'_i, z_i) \geq 0  )) \\
& \qquad \vee (x'_i {=} z_i) \big) \implies \\
& \Big( \underset{ \substack{q \in Q \cap Q_{Acc} \\ q' \in \delta(q, \lab(\bm{x}))}}{\bigwedge}  \big( \barrier_{q'}(\bm{x}', \bm{z})  - \tau_0 \barrier_q(\bm{x}, \bm{z}  ) - \\
& \qquad \sum_{i = 1}^b (x_{v+i} - \sigma_i(\bm{x}_v, \bm{z}) ) g_{2v}( \bm{x}, \bm{z}) \geq 0) \big)  \Big),    
\end{align*}
Now, we split the disjunction in the antecedent to a set of conditions: one for each of $((\transclos(x_i, x'_i) \geq 0) \wedge ( \transclos (x'_i, z_i) \geq 0  ))$, $(x'_i = z_i) $, and $(x'_i = x_i) $, respectively, as: 
\begin{align*}
&\big(\bm{g}(\bm{x}) \geq 0 \big) \wedge \big( g_v(\bm{z}) \geq 0 \big) \wedge \big( \underset{1 \leq i \leq v}{\bigwedge} (x'_i \in f(x_i) ) \big) \wedge \big( V1 ) \big)  \\
& \implies \Big( \underset{ \substack{q \in Q \cap Q_{Acc} \\ q' \in \delta(q, \lab(\bm{x}))}}{\bigwedge}  \big( \barrier_{q'}(\bm{x}', \bm{z})  - \tau_0 \barrier_q(\bm{x}, \bm{z}  ) - \\
& \qquad \sum_{i = 1}^b (x_{v+i} - \sigma_i(\bm{x}_v, \bm{z}) ) g_{2v}( \bm{x}, \bm{z}) \geq 0) \big)  \Big),    
\end{align*}
where condition $V1$ is replaced by $((\transclos(x_i, x'_i) \geq 0) \wedge ( \transclos (x'_i, z_i) \geq 0  ))$, or $(x'_i = z_i) $, for each $1 \leq i \leq v$, respectively.
We adopt a similar strategy for nondeterminism: that is, we assume that there are at most $n_d$ choices for $x'_i \in f(x_i)$, that is, $|f(x_i)| \leq n_d$ for any $x_i \in \states$.
Thus, this condition above results in $(2+n_d)^v$ conditions in the SOS-formulation.
We should note that the choice of functions $\sigma_i$ must be the same SOS polynomial for all states $q \in Q$ as the choice of the next state of the existential player cannot depend on the universal player.
Now, that all conditions are of the form $\underset{1 \leq i \leq t}{\bigwedge}(\varsigma_i \geq 0) \implies (\varsigma_0 \geq 0) $ where $\varsigma_i$ are polynomials, we can use Putinar's positivstellensatz \cite{putinar1993positive} to reduce the difficulty of finding a HyperCertificate as in Definition \ref{def:Hyper_cert_FE} to synthesize appropriate SOS-polynomials similar to \cite{chatterjee2024sound}.
For ease of presentation, we present the details of the SOS formulation of HyperCertificates in the Appendix \ref{ap:SOS}.

\subsection{SMT-based CEGIS Approach}
\label{subsec:cegis}
Similarly to \cite{murali2024closure}, one may instead  adopt a counterexample-guided inductive synthesis (CEGIS) framework to synthesize the certificate functions $\transclos$ and $\barrier$. 
To do so, similar to the SOS-approach, we assume that the transition function is a polynomial in the state variable $x$, parameterize functions $\transclos(x,y) = \sum_{l=0}^{l =z_1}t_l,\chi_l(x,y)$, and consider piece-wise functions $\barrier_q(\bm{x}, \bm{z}) = \sum_{l_q=0}^{l_q = z_2}c_{l,q}\upsilon_{q,l}(\bm{x}, \bm{z})$ for each automaton state $q \in Q$ where $t_l,c_{l,q}\in\mathbb{R}$ are unknown coefficients and $\chi_l: \states \times \states$ and $\upsilon_q: \prodstate \times \states^v $ are monomials to act as HyperCertificates. 
We then sample a finite number of points and reduce the check of a certificate over these points to solving a bilinear program over these sampled points. 
We then try to falsify the above conditions in an SMT solver such as z3 \cite{moura_2008_z3}.
If we can falsify the conditions, then the SMT solver provides a model with a corresponding state that violates at least one of conditions  \eqref{eq:cc_base}, \eqref{eq:cc_inv}, \eqref{eq:Hyper_Cert_base}, \eqref{eq:Hyper_Cert_accept}, or\eqref{eq:Hyper_Cert_rank}.
We add the corresponding state to the respective condition and solve the bilinear program again.
We repeat this iteration until we fail to falsify our candidate HyperCertificate.

\section{Case Studies}
\label{sec:case_studies}
We evaluate our proposed framework on two discrete-time dynamical systems to demonstrate its effectiveness in certifying initial-state opacity using the synthesized functions $\transclos$ and $\barrier$.

\subsection{One-Dimensional Numerical System}

We consider a one-dimensional system with non-deterministic transitions defined as $\system = (\states, \initstates, \transfunc)$, where the state space is $\states = \Reals_{\geq 0}$ and the initial set is $\initstates = \{5, 6\}$. The transition function $\transfunc$ is given by:
\[
\transfunc(x) =
\begin{cases}
\lfloor x/2 \rfloor + 1, & \text{if } x \in \{5,6\}, \\
\lfloor x/2 \rfloor,     & \text{if } x \in \{5,6\}, \\
x/2,                     & \text{otherwise}.
\end{cases}
\]
The secret state set is $\secretstates = \{5\}$, and the non-secret states are $\states \setminus \secretstates$. Two states $(x_1, x_2)$ are indistinguishable to an external observer if $|x_1 - x_2| \leq 1$.

We apply our data-driven synthesis procedure using degree-2 polynomial templates for both $\transclos$ and $\barrier$. The resulting certificate functions are:
\begin{align*}
\transclos(x_1, x_3) = {} & 0.0028 + 0.0014 x_3 + 0.0028 x_1 + 0.0014 x_1 x_3 \\+ &0.0028 x_1^2, \\
\barrier(x_1, x_2, x_3) = {} & 0.0054 + 0.0016 x_3 + 0.0062 x_2 - 0.0006 x_2 x_3 \\ &- 0.0041 x_2^2 
                             + 0.0035 x_1 + 0.0004 x_1 x_3 +\\ &0.0008 x_1 x_2 + 0.0028 x_1^2,
\end{align*}
where $x_3$ denotes the value of the state of the system that may be reachable from state $x_1$, $x_1$ denotes the states that may be reachable from a secret initial state, and $x_2$ denotes the states that is reachable from an initial non-secret state.

\subsection{Vehicle Dynamics Model}
\label{subsec:case_studies_1}
We further evaluate our method on a three-dimensional vehicle model adapted from~\cite{anand2024verification}. The state is $(s, v, w)$, representing position, velocity, and acceleration, respectively. The position $s \in [0,8]$, velocity $v \in [0,0.6]$, and acceleration $w$ is chosen non-deterministically from $\{-0.04, -0.02, 0.0, 0.02, 0.04\}$. The system dynamics are:
\[
\begin{cases}
s' = s + v + \tfrac{1}{2} w, \\
v' = v + w, \\
w' \in \{-0.04, -0.02, 0.0, 0.02, 0.04\}.
\end{cases}
\]

The initial set is $\initstates = [0,8] \times [0,0.6] \times \{-0.04, -0.02, 0.0, 0.02, 0.04\}$, and the secret set is $\secretstates = [0,1] \times [0,0.6] \times \{-0.04, -0.02, 0.0, 0.02, 0.04\}$. The observer has access only to the velocity component, with finite precision $\delta = 0.01$. Thus, two states $(s_1, v_1, w_1)$ and $(s_2, v_2, w_2)$ are indistinguishable if $|v_1 - v_2| \leq 0.01$.
To synthesize a HyperCertificate as in Definitions \ref{def:acc_opac_1} and \ref{def:Hyper_cert_FE}, we first synthesize a template closure certificate $\transclos$. We then use this to synthesize a barrier certificate $\barrier$ rather than trying to synthesize both simultaneously.
We find a closure certificate and a barrier certificate of degree $3$ in $4$ hours on a machine with $64$GB of RAM running on an Intel i$7$-$12700$KF.
The HyperCertificate is presented in Appendix \ref{ap:case_studies_1}. 

\subsection{Finite-State Example}
\label{subsec:case_studies_2}
We evaluate our approach on the finite-state system illustrated in Figure~\ref{fig:abc_fails}, where the augmented barrier certificate method proposed in~\cite{anand2024verification} fails to certify opacity, as discussed in Section~\ref{sec:intro}. 

Since the reachable set of states is finite and explicitly known, we do not require the synthesis a closed form function $\transclos$. 
Instead, we focus on synthesizing a barrier function $\barrier$ using a sum-of-squares (SOS) optimization framework over a fixed polynomial template of degree $3$.
The certificate is present in Appendix \ref{ap:case_studies_2} and takes $5$ minutes on the reference machine.

\section{Conclusion}
\label{sec:concl}
We presented a notion of HyperCertificates to verify discrete-time dynamical systems against hyperproperties specified in HyperLTL.
Such certificates allow one to encode lookahead information and thus are more general than existing barrier certificate-based approaches to verify hyperproperties.
As future work, we plan to investigate whether one may be able to find not just sufficient but also necessary conditions for certain classes of systems and properties.
Another direction we leave for future work is to investigate how to reason about stochastic systems and to design correct-by-construction controllers to ensure the satisfaction of hyperproperties as well as the use of neural networks to represent such certificates.
We also plan to investigate the use of relational abstractions \cite{tamir2023counterexample} as well as prophecy variables \cite{finkbeiner2026universal} to reduce the burden of the existential quantifiers in our certificates.

\bibliographystyle{alpha}
\bibliography{references}
\section{Appendix}
\section{SOS-formulation of HyperCertificates}
\label{ap:SOS}
Functions $\transclos: \states \times \states \to \Reals$, and $\barrier_q: \prodstate \times \states^v \to \Reals$ are a HyperCertificate as in Definition \ref{def:Hyper_cert_FE} if the following polynomials are SOS:
\begin{align*}
& \transclos(x,f(x))-\lambda_{0}^T(x)g(x),\\
&\transclos(x,y) - \tau_1\transclos(f(x),y)-\lambda_{1}^T(x,y)g_{2}(x,y), \\
\end{align*}
where $\lambda_{0}, \lambda_1$ are SOS-polynomials and 
\begin{align*}
&\big( \underset{ \substack{q_0 \in Q_0}}{\bigwedge}  ( \barrier_{q_0}(\bm{x}, \bm{z}) - \sigma_{0,0}(\bm{x}) \bm{g}_0(\bm{x})  \\
& \qquad - \sum_{i = 1}^b (x_{v+i} - \sigma_{0,1,i}(\bm{x}_v, \bm{z}) ) g_{2v}(\bm{x}_v, \bm{z})  \\
& - \sigma_{0,2}(\bm{z}) g_v(\bm{z}) - (\bar{\transclos}(\bm{x}_v, \bm{z})- \sigma_{0,3}(\bm{x}_v, \bm{z}))
)  \big), \\
&\big( \underset{ \substack{q \in Q \cap Q_{Acc} \\ q' \in \delta(q, \lab(\bm{x}))}}{\bigwedge}  ( \barrier_{q'}(\bm{x}', \bm{z})  - \tau_0 \barrier_q(\bm{x}, \bm{z}  ) - \sigma_{1,0}(\bm{x}) \bm{g}(\bm{x})  \\
& \qquad - \sum_{i = 1}^b (x_{v+i} - \sigma_{1,1,i}(\bm{x}_v, \bm{z}) ) g_{2v}((\bm{x}_v, \bm{z}))  \\
& - \sigma_{1,2}(\bm{z}) g_v(\bm{z}) - (\bar{\transclos}(\bm{x}_v, \bm{z})- \sigma_{1,3}(\bm{x}_v, \bm{z}))
 \big), \\
 &\big( \underset{ \substack{q \in Q \setminus Q_{Acc} \\ q' \in \delta(q, \lab(\bm{x}))}}{\bigwedge}  ( -\barrier_{q'}(\bm{x}', \bm{z})  + \barrier_q(\bm{x}, \bm{z}  ) - \xi - \sigma_{2,0}(\bm{x}) \bm{g}(\bm{x})  \\
& \qquad - \sum_{i = 1}^b (x_{v+i} - \sigma_{2,1,i}(\bm{x}_v, \bm{z}) ) g_{2v}((\bm{x}_v, \bm{z}))  \\
& - \sigma_{2,2}(\bm{z}) g_v(\bm{z}) - (\bar{\transclos}(\bm{x}_v, \bm{z})- \sigma_{2,3}(\bm{x}_v, \bm{z}))
 \big), \\
\end{align*}
where $\sigma_i$ are SOS-polynomials and 
$\bm{x} = (x_1, \ldots, x_p) $ ,$\bm{x}_v = (x_1,\ldots, x_v)$, $\bm{x}_b = (x_{v+1}, \ldots, x_p) $, $\bm{z} = (z_1, \ldots, z_v)$ and $\bar{\transclos}(x,z) = \begin{bmatrix}
    \transclos(x_1, z_1) \\
    \vdots  \\
    \transclos(x_v, z_v)
\end{bmatrix}$.
Note that we do not write the conditions for the other disjuncts but they follow in a similar fashion.

\section{Certificates for Section \ref{subsec:case_studies_1}}
\label{ap:case_studies_1}
To find a HyperCertificate via SOS, we first fix all the bilinear variables introduced as $0.7$.
We then find a closure certificate $\transclos$ of degree $3$ to encode the lookahead as:

$\transclos(s_1, v_1, w_1, s_2, v_2, w_2) = 10.42650007326121s_1^3 - 21.54860310117765s_1^2v_1 - 1.4469345364567237 \cdot10^{-13}s_1^2s_2 + 2.977913297673249\cdot 10^{-12}s_1^2v_2 + 10.737604265802078s_1v_1^2 + 9.226696476061465\cdot 10^{-13}s_1v_1s_2 + 2.7234880953761896\cdot 10^{-13}s_1v_1v_2 + 122.33828608326392s_1s_2^2 - 43.97489181259062s_1s_2v_2 + 157.44807863391816s_1v_2^2 + 32.15010553145641v_1^3 - 1.248905975189414\cdot 10^{-12}v_1^2s_2 + 1.6799596999233592\cdot 10^{-12}v_1^2v_2 + 110.4487673809787v_1s_2^2 - 62.25703038659162v_1s_2v_2 + 157.99270121443507v_1v_2^2 - 4.91103699474163\cdot 10^{-14}s_2^3 - 1.140176430624521\cdot 10^{-13}s_2^2v_2 + 7.56539824866095\cdot 10^{-14}s_2v_2^2 + 9.711975097029537\cdot 10^{-14}v_2^3 + 23.57832640911535s_1^2 + 12.499071923340196s_1v_1 - 2.427967561124748\cdot 10^{-12}s_1s_2 + 1.3214277859142432\cdot 10^{-12}s_1v_2 + 90.04012900502701v^2 - 2.257591407727963\cdot 10^{-12}v_1s_2 + 1.0288597619648384\cdot 10^{-12}vv_2 + 235.75446721733394s_2^2 - 83.38002053098964s_2v_2 + 304.11679345617307v_2^2 + 54.19531613209792s_1 + 82.64012621498918v_1 - 2.0689705198013913e-12s_2 + 1.9874667353892596e-12v_2 + 103.87025070937855$.

We fix the above closure certificate, and then search for a function $\barrier$ of degree $3$.
As we use $w$ to denote the nondeterminism (either the existential or universal choice of the next state), we consider the function $\barrier(s_1,s_2, v_1,v_2, s_{lookahead}, v_{lookahead})$ and do not parameterize it with respect to $w$:
We find a degree $3$ certificate as follows:

$\barrier(s_1,s_2, v_1, v_2, s_{lookahead}, v_{lookahead}) =  0.0003 s_1^3 - 0.0026 s_1^2 v_1 + 0.0009 s_1^2 v_2 + 0.0049 s_1 v_1^2 - 0.0002 s_1 v_1 s_2 - 0.0040 s_1 v_1 v_2 - 0.0007 s_1 s_2 v_2 + 0.0034 s_1 v_2^2 - 0.0001 s_1 s_{lookahead} v_{lookahead} + 0.0004 v_1^3 + 0.0005 v_1^2 s_2 - 0.0003 v_1^2 v_2 - 0.0003 v_1 s_2^2 + 0.0015 v_1 s_2 v_2 - 0.0008 v_1 v_2^2  + 0.0001 v_1 s_{lookahead}^2 + 0.0002 v_1 s_{lookahead} v_{lookahead} + 0.0001 v_1 v_{lookahead}^2 + 0.0001 s_2^3 - 0.0004 s_2^2 v_2 - 0.0001 v_2 s_{lookahead} v_{lookahead} - 0.0028 s_1^2 + 0.0137 s_1 v_1 - 0.0024 s_1 v_2 - 0.0229 v_1^2 - 0.0021 v_1 s_2 + 0.0176 v_1 v_2 - 0.0011 s_2^2 + 0.0085 s_2 v_2 - 0.0173 v_2^2 + 1.3017\cdot 10^{-5} s_{lookahead} v_{lookahead} - 7.9592\cdot 10^{-6} v_{lookahead}^2 + 0.0025 s_1 - 0.0037 v_1 + 0.0003 s_2 - 0.0010 v_2 + 0.009$,
where $s_{lookahead}$ and $v_{lookahead}$ denotes the future states that may be reached from states $s_1$ and $v_1$, the states $s_1$, $v_1$ denote states that may be reachable from an initial secret state, and $s_2$, $v_2$ denote the states that are reachable from a non-secret initial state.

\section{Certificates for Section \ref{subsec:case_studies_2}}
\label{ap:case_studies_2}
As we assume the closure certificate is explicitly provided for the finite-state example, we only synthesize a barrier certificate. We find a degree three certificate for function $\barrier$ as :
\begin{align*}
\barrier(x_1,x_2,x_3) = {} & 0.04165 x_{1}^{3} - 0.14805 x_{1}^{2} x_{2} + 0.32226 x_{1}^{2} x_{3} \\
                           & + 0.39294 x_{1} x_{2}^{2} + 0.56647 x_{1} x_{2} x_{3} - 0.17507 x_{1} x_{3}^{2} \\
                           & - 0.38673 x_{2}^{3} + 0.38294 x_{2}^{2} x_{3} - 0.32596 x_{2} x_{3}^{2} -\\& 0.02524 x_{3}^{3} 
                            - 1.48182 x_{1}^{2} - 3.34472 x_{1} x_{2} \\&- 2.60942 x_{1} x_{3} 
                            + 3.21638 x_{2}^{2} - 1.00687 x_{2} x_{3} \\&- 1.37964 x_{3}^{2} 
                            + 0.7789 x_{1} - 0.20299 x_{2} \\&- 0.96456 x_{3} - 0.59348
\end{align*}
where $x_3$ denotes the lookahead, $x_1$ the states that may be reachable from the secret state, and $x_2$ the states from the non-secret state.

\end{document}